\documentclass[aps,pra,notitlepage,10pt,twocolumn,superscriptaddress,floatfix]{revtex4-2}
\pdfoutput=1
\usepackage[utf8]{inputenc}
\usepackage[english]{babel}
\usepackage[T1]{fontenc}
\usepackage{amsmath}
\usepackage{amsfonts}
\usepackage{amssymb}
\usepackage{amsthm}
\usepackage{makeidx}
\usepackage{graphicx}
\usepackage{color}
\usepackage{xspace}
\usepackage{thmtools, thm-restate}
\usepackage[normalem]{ulem}
\definecolor{mpBlue}{RGB}{21, 101, 192}
\definecolor{mpGreen}{RGB}{46, 125, 50}
\usepackage{hyperref}
\hypersetup{citecolor=mpGreen}
\hypersetup{colorlinks=true}
\hypersetup{linkcolor=mpBlue}
\hypersetup{urlcolor=mpBlue}
\newtheorem*{proposition*}{Proposition}
\newtheorem*{theorem*}{Theorem}
\newtheorem*{lemma*}{Lemma}
\declaretheorem{proposition}
\declaretheorem{theorem}
\declaretheorem[sibling=proposition]{lemma}
\newtheorem{definition}[theorem]{Definition}
\newtheorem{corollary}[theorem]{Corollary}
\renewcommand{\paragraph}[1]{\addcontentsline{toc}{section}{#1}\emph{#1.}---}
\DeclareMathOperator{\NN}{\mathbb{N}}

\DeclareMathOperator{\CC}{\mathbb{C}}
\DeclareMathOperator{\id}{id}
\usepackage{dsfont}
\DeclareMathOperator{\I}{\mathds{1}}
\DeclareMathOperator{\lin}{\mathcal{L}}
\DeclareMathOperator{\dens}{\mathcal{D}}
\DeclareMathOperator{\linspan}{span}
\DeclareMathOperator{\Tr}{Tr}
\DeclareMathOperator{\oM}{\mathsf{M}}
\DeclareMathOperator{\oP}{\mathsf{P}}
\DeclareMathOperator{\oT}{\mathsf{T}}
\DeclareMathOperator{\J}{\mathcal{J}}

\newcommand*{\bra}[1]{\langle #1 |}
\makeatletter
\newcommand*{\ket}[1]{| #1 \rangle \@ifnextchar\bra{\!}{}}
\makeatother

\newcommand*{\ketbra}[1]{| #1 \rangle \! \langle #1 |}
\newcommand{\etaJM}{\eta^{\mathrm{JM}}}
\newcommand{\etaB}[1]{\eta^{\mathrm{Bell}}_{#1}}

\begin{document}

\title{All incompatible measurements on qubits lead to multiparticle Bell nonlocality}
\author{Martin Pl\'{a}vala}

\affiliation{Naturwissenschaftlich-Technische Fakult\"{a}t, Universit\"{a}t Siegen, Walter-Flex-Stra\ss e 3, 57068 Siegen, Germany}
\affiliation{Institut f\"{u}r Theoretische Physik, Leibniz Universit\"{a}t Hannover, Hannover, Germany}

\author{Otfried G\"{u}hne}
\affiliation{Naturwissenschaftlich-Technische Fakult\"{a}t, Universit\"{a}t Siegen, Walter-Flex-Stra\ss e 3, 57068 Siegen, Germany}

\author{Marco Túlio Quintino}
\affiliation{Sorbonne Université, CNRS, LIP6, Paris F-75005, France}

\begin{abstract}
Bell nonlocality is a fundamental phenomenon of quantum physics as well as an essential resource for various tasks in quantum information processing. It is known that for the observation of nonlocality the measurements on a quantum system have to be incompatible, but the question which incompatible measurements are useful, remained open. Here we prove that any set of incompatible measurements on qubits leads to a violation of a suitable Bell inequality in a multiparticle scenario, where all parties perform the same set of measurements. Since there exists incompatible measurements on qubits which do not lead to Bell nonlocality for two particles,  our results demonstrate a fundamental difference between two-particle and multi-particle nonlocality, pointing at the superactivation of measurement incompatibility as a resource. In addition, our results imply that measurement incompatibility for qubits can always be certified in a device-independent manner.
\end{abstract}

\maketitle

\paragraph{Introduction}
The study of quantum correlations in the form of entanglement \cite{horodecki2009entanglement, GuhneToth-entanglement}, steering \cite{WisemanDohertyJones-nonlocal, UolaCostaNguyenGuhne-steering}, and Bell nonlocality \cite{BrunnerCavalcantiPironioScaraniWehner-BellNonlocality} has been a central topic in quantum mechanics for several decades, with Bell nonlocality being the strongest form of correlation among them. In additional to its relevance to quantum foundations~\cite{Pawlowski09Causality}, Bell nonlocality plays a deep role in various branches of quantum information and quantum computation~\cite{Cleve04NLgames,MPIstar,cabello23_MIPstar}. The discovery of device-independent protocols~\cite{Pironio2009DI} allows Bell nonlocality to be used for proving security in cryptographic protocols where the attackers may access general physical theories beyond quantum physics \cite{Zhang_2022,Nadlinger_2022}.

In quantum theory, Bell nonlocality requires the combination of two key ingredients, entangled states and incompatible measurements. Studies of the relationship between entanglement and Bell nonlocality date back to 1989, with Werner's seminal paper~\cite{Werner1989}, which provides an explicit local hidden variable (LHV) model for quantum states, showing that entanglement is not sufficient for Bell nonlocality with projective measurements. This result was then extended to treat the most general class of measurements, described by a Positive Operator-Valued Measure (POVM)~\cite{Barrett2002POVM}. Our knowledge on how entanglement relates to nonlocality has been growing since then~\cite{almeida07,augusiak_review,quintino15EntSteBell,hirsch15,cavalcanti15}, leading to the phenomena of entangled states with hidden nonlocality~\cite{popescu95,hirsch13}, entangled states whose nonlocality can be superactivated by performing joint measurements on multiple copies~\cite{palazuelos12,Quintino16steeringSuperActivation}, among others~\cite{sen03,cavalcanti10}.

When compared to entanglement, the relationship between measurement incompatibility and Bell nonlocality is less understood. A first striking result came in 2009~\cite{WolfPerezgarciaFernandez-measIncomp}, where Wolf, Perez-Garcia, and Fernandez showed that a pair of two-outcome POVMs can be used to violate a Bell inequality, if and only if it is incompatible. Later, it was shown that measurement incompatibility is necessary and sufficient for steering~\cite{QuintinoVertesiBrunner-steering,UolaMoroderGuhne-steering}. However, for general bipartite scenarios, it was proven that there are sets of measurements which are incompatible, but cannot lead to Bell nonlocality~\cite{quintino15BellJM,HirschQuintinoBrunner-noBell,BeneVertesi-noBell}, a result which may be viewed as the existence of LHV models for measurements, in similar spirit to entangled, 
but Bell-local Werner states~\cite{Werner1989}.

In this paper, we address whether incompatible POVMs can violate Bell inequalities and lead to experimental demonstrations of quantum nonlocality. Contrary to the bipartite case, we show that for qubits any set of incompatible POVMs leads to Bell nonlocality and violates some multipartite Bell inequality where all parties perform the same set of measurements, see Fig.~\ref{fig:jm-nl}.

\begin{figure}
\centering
\includegraphics[width=\linewidth]{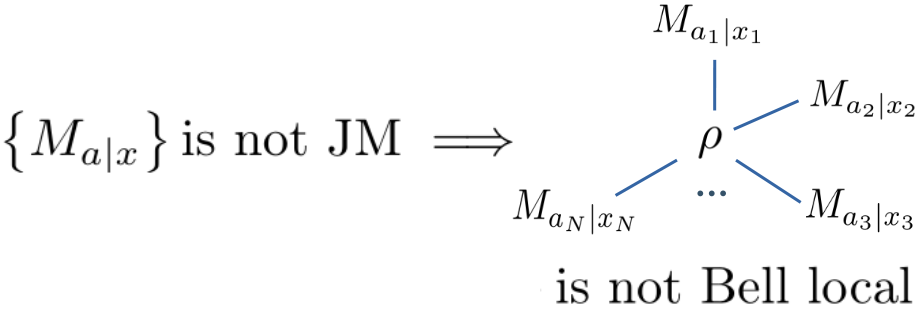}
\caption{In this paper we show that, if $\{M_{a|x}\}$ is a set of qubit measurements (so-called POVMs where $a$ and $x$ label the inputs and outputs, respectively) which is incompatible, then there exist a number of parties $N \in \NN$ and a quantum state $\rho\in\dens(\CC_2^{\otimes N})$ such that the behaviour $ p(a_1 \ldots a_N| x_1 \ldots x_N) = \Tr[ \rho (M_{a_1|x_1} \otimes \ldots \otimes M_{a_N|x_N})]$ is Bell nonlocal, i.e., it violates some Bell inequality.}
\label{fig:jm-nl}
\end{figure}

\paragraph{Preliminaries}
We will use $\CC_d$ to denote a complex, finite-dimensional Hilbert space, $\dim(\CC_d) = d$, $\lin(\CC_d)$ will denote the set of self-adjoint (linear) operators on $\CC_d$ and an operator $A \in \lin(\CC_d)$ is positive semidefinite, $A \geq 0$, when $A$ is self-adjoint and all the eigenvalues of $A$ are non-negative. $\dens(\CC_d)$ will denote the set of density matrices which represent states of a quantum system associated to the Hilbert space $\CC_d$, that is, $\rho \in \lin(\CC_d)$ if $\rho \geq 0$ and $\Tr(\rho) = 1$ and $\I$ will denote the identity operator.

A POVM  $\{ M_{a} \}$  is a set of positive semidefinite operators that sum to identity, i.e., $M_a \geq 0$ and $\sum_a M_a = \I$. POVMs represent the most general description of a measurement that can be performed on a quantum system, and the probability of obtaining the outcome $a$ when performing the POVM $\{M_a\}$ on the state $\rho \in \dens(\CC_d)$ is given by the Born's rule, which reads as $p(a) = \Tr(\rho M_a)$.

A set of POVMs is a set of positive semidefinite operators $M_{a|x}$ such that for every fixed value of $x$, $\{M_{a|x}\}$ is a POVM, in this way $x$ may be viewed as input choice and $a$ as the outcome. We will use $\oM = \{ M_{a|x} \}$ to denote the whole set, while $M_{a|x}$ will denote the specific operator. A set of POVMs represents several possible measurements, and one can ask whether all the measurements in the collection can be measured in a single experiment. This is captured by the definition of joint measurability \cite{HeinosaariMiyaderaZiman-compatibility,GuhneHaapasaloKraftPellonpaaUola-incomaptiblity}.

\begin{definition}
 A set of POVMs $\{M_{a|x}\}$ is jointly measurable  (JM) or compatible if it can be decomposed as
\begin{align}
    M_{a|x}=\sum_{\lambda} p(a|x, \lambda) G_\lambda, \quad \forall a,x,
\end{align}    
where $\{G_\lambda\}$ is the so-called mother POVM, obeying $G_\lambda \geq 0$ and $\sum_\lambda G_\lambda = \I$, and the $p(a|x,\lambda)$ are conditional probability distributions, i.e., $p(a|x,\lambda)\geq0$, and $\sum_ap(a|x,\lambda)=1$ for any $x,\lambda$. Set of measurements is called incompatible if it is not JM.
\end{definition}

Quantum systems are composed by means of the tensor product, a quantum state $\rho \in \dens(\CC_d^{\otimes N})$ may then be viewed as a quantum state shared by $N$ parties which we label with an integer $n\in\{1,\ldots,N\}$. We now consider a multipartite scenario where each party may perform quantum measurements on their respective local systems. These measurements may be described by the sets of POVMs $\oM^n = \{ M^n_{a_n|x_n} \}$ and the probability that the parties obtain the respective outcomes $a_1 \ldots a_N$ given the choices of measurements $x_1 \ldots x_N$ is $p(a_1 \ldots a_N| x_1 \ldots x_N) = \Tr[ \rho (M^1_{a_1|x_1} \otimes \ldots \otimes M^n_{a_n|x_n})]$. As standard in the literature of Bell nonlocality,  we refer to a set of probability distributions $\{p(a_1 \ldots a_N| x_1 \ldots x_N) \}$ as a (multipartite) behaviour. For such a behaviour, one can ask whether it is nonlocal or not \cite{BrunnerCavalcantiPironioScaraniWehner-BellNonlocality}.

\begin{definition}\label{def:Bell}
A behaviour $\{p(a_1 \ldots a_N| x_1 \ldots x_N)\}$ is (N-partite) Bell local if it can be written as 
\small
\begin{equation}
    p(a_1 \ldots a_N| x_1 \ldots x_N) = \sum_\lambda p(\lambda) p_1(a_1|x_1, \lambda) \ldots p_N(a_N|x_N, \lambda)
\end{equation}
\normalsize
where $\{p(\lambda)\}$ is a probability distribution and $p_n(a_n|x_n, \lambda)$ are conditional probabilities.
\end{definition}

\paragraph{Bell nonlocality and joint measurability}
The relationship between measurement incompatibility and Bell nonlocal correlations have been investigated over the years from foundational and practical perspectives \cite{anderson05,son05,WolfPerezgarciaFernandez-measIncomp,QuintinoVertesiBrunner-steering,UolaMoroderGuhne-steering,uola15}. It is straightforward to verify that measurement incompatibility is a necessary resource for Bell nonlocality, that is, if the set of measurements performed by $N-1$ parties are JM, the corresponding behaviour $p(a_1 \ldots a_N| x_1 \ldots x_N) = \Tr[ \rho_N (M^1_{a_1|x_1} \otimes \ldots \otimes M^n_{a_n|x_n})]$ is Bell local regardless of the quantum state $\rho$ and regardless of the measurements performed by the other party~\cite{Fine-BellIneq}. From a more practical perspective, the necessity of measurement incompatibility for Bell nonlocality allows a device-independent certification of measurement incompatibility, that is, if  a behaviour $p(a_1 \ldots a_N| x_1 \ldots x_N)$ is Bell nonlocal, at least two parties are performing incompatible measurements \cite{chen2016natural,quintino2019device}.

The question whether measurement incompatibility is a sufficient requirement for Bell nonlocality is considerably more complex. Refs.~\cite{WolfPerezgarciaFernandez-measIncomp,loulidi2022measurement} shows that a pair of two-outcome POVMs $\oM=\{M_{0|0}, M_{1|0}, M_{0|1}, M_{1|1}\}$ can be used to violate a Bell inequality, if and only if $\oM$ is incompatible. However, for bipartite scenarios, this strong relationship between Bell nonlocality and measurement incompatibility turns out to be a particularity of a pair of two-outcome measurements. Refs~\cite{HirschQuintinoBrunner-noBell,BeneVertesi-noBell} present explicit sets of incompatible qubit measurements $\{M_{a|x}\}$ such that, for any bipartite state, and any set of quantum measurements for Bob $\{N_{b|y}\}$ the behaviour given by $p(ab|xy)=\Tr[\rho_{AB} (M_{a|x} \otimes N_{b|y})]$ is guaranteed to be Bell local (see also Ref.~\cite{quintino15BellJM}, which proves the same result under the hypothesis that Bob performs two-outcome POVMs). Hence, in a bipartite scenario, even for qubits measurement incompatibility is not a sufficient resource to observe Bell nonlocality.

\paragraph{Main result}
We now state our main result.

\begin{restatable}{theorem}{thmMain} \label{thm:main}
Let $\CC_2$ be the qubit Hilbert space, $\dim(\CC_2) = 2$, and let $\oM = \{M_{a|x}\}$ be a finite set of qubit POVMs, that is $M_{a|x} \in \lin(\CC_2)$, $M_{a|x} \geq 0$ and $\sum_{a} M_{a|x} = \I$; $a \in \{1, \ldots, n_a\}$, $x \in \{1, \ldots, n_x\}$ and $n_a, n_x \in \NN$. If $\oM$ is incompatible, then there exists an $N \in \NN$ and a state $\rho \in \dens(\CC_2^{\otimes N})$ such that the behaviour $p(a_1 \ldots a_N| x_1 \ldots x_N) = \Tr[\rho (M_{a_1|x_1} \otimes \ldots \otimes M_{a_N|x_N})]$ is multipartite Bell nonlocal.
\end{restatable}

In short, this theorem states that, for qubits, if a set of POVMs $\oM$ is incompatible, there exists a quantum state $\rho$ such that if the $N$ parties perform exactly the same measurements $\oM$ on $\rho$, the resulting behaviour leads to a violation of some Bell inequality. The detailed proof of Theorem~\ref{thm:main} is delegated to the Appendices \ref{app:obs1}, \ref{app:obs2}, \ref{app:obs3}, and \ref{app:proof}, it relies on the following four key steps.

(a) First, we reinterpret the notion of Bell nonlocality of a behaviour. In fact, Bell locality in the sense of Def.~\ref{def:Bell} means that the behaviour can be seen as separable state in the tensor product space of conditional probabilities. This is not the standard notion of separability of quantum states \cite{GuhneToth-entanglement}, rather it is a notion of separability for generalized probabilistic theories \cite{Barrett-GPTinformation,Plavala-review}.

(b) Then, we consider the set of POVMs $\oM = \{M_{a|x}\}$ containing $n_x$ measurements. Since $p(a|x) = \Tr(\rho M_{a|x})$ is a conditional probability, we can consider $\oM$ as a map from density matrices to conditional probabilities. Ideas similar to the Choi-Jamio{\l}kowski isomorphism \cite{dePillis1967linear,jamiolkowski1972linear,choi1975completely} associate to this map a state of the tensor product of positive operators and conditional probabilities. We show that if $\oM$ as a map does not correspond to a separable state in this picture, then there must exist a certain positive, but not completely positive map $\Phi$ on the space of quantum states, arising from a witness-like construction.

(c) Third, we consider general maps on qubits. We prove that a positive map that is entanglement annihilating must be completely positive, which is an extension of a result in Ref.~\cite{ChristandlMullerhermesWolf-composedETB}. This connects to (a) and (b) as follows: If a set of POVMs $\oM$ applied to $N$ qubits gives always Bell-local, i.e., separable according to (a), behaviours, then maps like $\Phi$ must be entanglement annihilating. It follows that in this case the map $\oM$ in the sense of (b) must be separable, otherwise there is a contradiction to $\Phi$ not being completely positive.

(d) Finally, we close the argument. We show that if $\oM$ is a separable map in (b), then the set of POVMs is JM. This follows by interpreting the separable decomposition in (b) as a mother POVM for all the POVMs in $\oM$, similar arguments are known \cite{NamiokaPhelps-cones, Jencova-incomaptibility,Jencova-steering}. This proves the Theorem: If the POVMs in the set $\oM$ are incompatible, then there must be an $N$ where the behaviour generated by these measurements on qubits becomes Bell nonlocal. If this were not the case, then according to (c), $\oM$ is a separable map in (b) and we would arrive at a contradiction.

An immediate question to ask is whether for a given set of POVMs $M$ one can identify necessary or sufficient number of parties $N$ for a violation of some Bell inequality. The most feasible way to obtain analytic bounds would be to investigate whether the positive but not completely positive map constructed in step (b) is locally entanglement annihilating for $N$ parties: if the map $\Phi$ is not locally entanglement annihilating, then also $M$ is not locally entanglement annihilating and thus some Bell inequality must be violated. For $N=2$ such conditions are known for specific classes of positive maps \cite{filippov2012local,filippov2013bipartite}, so they are not directly applicable to our case. In general deciding whether a given map is locally entanglement annihilating is related to whether a related map is tensor stable positive \cite{MullerhermesReebWolf-tensorPositive}, which was recently conjectured to be undecidable problem \cite{van2022halos}. We will later show that this problem may be approached numerically for certain scenarios.

It is natural to ask whether our result can be extended to $d>2$ dimensions. In the following we prove that, if our result does not hold for an arbitrary dimension $d$, then there exist bipartite qudit states $\rho_{AB}\in\mathcal{D}(\mathbb{C}_{d^2})$ which has a non-positive partial transpose (NPT), i.e., $\rho_{AB}^{T_B}$ is not positive, and $\rho_{AB}$ is bound entangled, i.e., a maximally entangled state cannot be distilled from many copies of $\rho_{AB}$~\cite{Horoceki1998BoundEnt}. The existence of NPT bound entangled states is a longstanding important open problem in quantum information, which has been studied from various researchers and different perspectives~\cite{IQOQIopen,Horodecki2022Five,bennett1999unextendible,Watrous2004Bound,Vianna2006Bound,MullerhermesReebWolf-tensorPositive,AubrunMullerhermes-anihilation}. Hence, we expect that generalising our main theorem for $d>2$ to be a very hard mathematical problem.
\begin{restatable}{theorem}{thmNPTboundEnt} \label{thm:NPTboundEnt}
If there exist a set of $d$-dimensional incompatible POVMs $\oM = \{M_{a|x}\}$ such that for all $N \in \NN$ and all states $\rho \in \dens(\CC_d^{\otimes N})$, the behaviour $p(a_1 \ldots a_N| x_1 \ldots x_N) = \Tr[\rho (M_{a_1|x_1} \otimes \ldots \otimes M_{a_N|x_N})]$ is Bell local, then there exist an NPT bound entangled bipartite quantum state in dimension $d^2$.
\end{restatable}
If the conditions of the Theorem~\ref{thm:NPTboundEnt} hold, then one can use the ideas from the proof of Theorem~\ref{thm:main} to show that there must be maps which are entanglement annihilating, but not entanglement breaking, which implies the existence of NPT bound entanglement according to Ref.~\cite{AubrunMullerhermes-anihilation}; details can be found in the Appendix~\ref{app:proofNPTboundEnt}.

\paragraph{Superactivation of measurement incompatibility for Bell nonlocality}
As previously mentioned, Refs.~\cite{HirschQuintinoBrunner-noBell, BeneVertesi-noBell} present sets of qubit dichotomic measurements $\oM = \{ M_{a|x} \}$ which are not jointly measurable, but, for any quantum state $\rho_{AB} \in \dens(\CC_2 \otimes \CC_d)$ and any sets of measurements $\{ B_{b|y} \}$, the bipartite behaviour $p(ab|xy) = \Tr[\rho_{AB} (M_{a|x} \otimes B_{b|y})]$ is Bell local. This result is interpreted as existence of incompatible measurements which are useless for Bell nonlocality in bipartite scenarios. When multipartite scenarios are considered, our main result shows that  all sets of incompatible qubit measurements can lead to Bell nonlocality. One might then argue that the sets of incompatible but Bell local measurements $\oM = \{ M_{a|x} \}$ presented in Refs.~\cite{HirschQuintinoBrunner-noBell,BeneVertesi-noBell} have some nonlocality that is activated in a scenario with multiple parties.

Interestingly, Ref.~\cite{HirschQuintinoBrunner-noBell} shows that, if $\oM = \{ M_{a|x} \}$ is a set of measurements leading to Bell local correlations on any bipartite scenario, this set of measurements cannot lead to genuine multipartite Bell nonlocality if more parties are considered. Let us, for concreteness, present the argument for the tripartite case. Let $p(abc|xyz)$ be the probabilities of a tripartite behaviour. For the given measurements, it must be by assumption Bell local in the $A/BC$ bipartition. That is, it can be written as $p(abc|xyz)=\sum_\lambda p(\lambda)p_A(a|x,\lambda) p_{BC}(bc|yz,\lambda)$. There exists various non-equivalent notions of genuine multipartite Bell nonlocality \cite{Bancal2011}, but the loosest one is given by Svetlichny~\cite{Svetlichny1987}, and states that a tripartite behaviour is not genuine tripartite Bell local if it can be written as a convex combination of behaviours which are Bell local in a bipartition. Hence, if a behaviour is Bell local in the $A/BC$ bipartition, it cannot be genuine multipartite nonlocal.

In combination with out results this has interesting consequences. Let $\oM = \{ M_{a|x} \}$ be the set of qubit measurements presented in Refs.~\cite{HirschQuintinoBrunner-noBell,BeneVertesi-noBell}. For any $N \in \NN$ and any quantum state $\rho\in \dens(\CC_2^{\otimes N})$, the behaviour $p(a_1 \ldots a_N| x_1 \ldots x_N) = \Tr[\rho (M_{a_1|x_1} \otimes \ldots \otimes M_{a_n|x_n})]$ is Bell local in every bipartition of the form $1/N-1$, that is, where one side of the partition has a single party. However, it follows from our main result that, there exists some $N \in \NN$, and a quantum state  $\rho\in \dens(\CC_2^{\otimes N})$ such that $p(a_1 \ldots a_N| x_1 \ldots x_N) = \Tr[\rho (M_{a_1|x_1} \otimes \ldots \otimes M_{a_n|x_n})]$ is Bell nonlocal, but Bell local in all $1/N-1$ partitions, hence not genuine multipartite nonlocal. Similar phenomenon has been phrased in previous literature as  \textit{anonymous nonlocality} \cite{Liang2014anonimous} and finds applications in multipartite secret sharing. Note that this phenomenenon also occurs in entanglement theory, where quantum states can be separable for any bipartition, but not fully separable \cite{bennett1999unextendible, acin2001classification, Vertesi2011MultiPartitePeres}.

\paragraph{Quantifying Bell incompatibility}
One may quantify measurement incompatibility by how robust it is against white noise~\cite{Busch-compatiblity,bavaresco2017_most_incompatible,DesignolleFarkasKaniewski-compatiblity,BluhmNechita-compatiblity} via the $\eta$ depolarising quantum channel, defined as $D_\eta(\rho) := \eta \rho + (1-\eta)  \Tr(\rho)  \I$. If $\oM=\{M_{a|x}\}$ is a set of POVMs, its noisy version is the POVM defined as $\oM^{(\eta)} := \left\{D_\eta(M_{a|x})\right\}$. The incompatibility robustness of a POVM $\oM$ is the maximum visibility parameter $\eta$ such that  $\oM^{(\eta)}$ is JM. More precisely, the maximum value of $\eta\in[0,1]$ such that $\oM^{(\eta)}$ is jointly measurable, i.e., $\eta^\text{JM}(\oM ):=\max \eta$ such that $\oM^{(\eta)}$ is JM. Since we set $\eta\in[0,1]$, a set of POVMs $\oM$ is jointly measurable if and only if $ \eta^\text{JM}(\oM )=1$.

Motivated by our main result, we introduce the definition of $k$-partite Bell joint measurability, which quantifies the critical visibility $\eta\in[0,1]$ for which a set of POVMs $\oM^{(\eta)}$ becomes Bell local in any $k$-partite scenario regardless of the shared state, where all parties perform the (same) measurements $\oM^{(\eta)}$.
\begin{definition}
\label{def:nBellJM}
Let $\oM=\{M_{a|x}\}$ be a set of POVMs acting in $\CC_d$. Its noisy $k$-partite Bell joint measurability $\etaB{k}(\oM)\in[0,1]$ is defined as $\etaB{k}(\oM):=\max \eta$ such that $\forall \rho \in \dens(\CC_d^{\otimes k})$, $\Tr[\rho (M^{(\eta)}_{a_1|x_1} \otimes \ldots \otimes M^{(\eta)}_{a_k|x_k})] $ is Bell local.
\end{definition}
Since only incompatible measurements can lead to Bell nonlocality, it follows that if $\oM$ is JM, we have $\etaB{k}(\oM) = 1$. Generally, $\eta_\text{JM}(\oM) \leq \etaB{k}(\oM)$ for any $k\in\NN$. Also, for any set of POVMs $\oM$ and any $k \in \NN$ it holds that $\etaB{k+1}(\oM) \leq \etaB{k}(\oM)$, this is true because if the marginal behaviour $\sum_{a_i} \Tr[\rho (M_{a_1|x_1} \otimes \ldots \otimes M_{a_k|x_k})]$ is Bell nonlocal, the full behaviour $\Tr[\rho (M_{a_1|x_1} \otimes \ldots \otimes M_{a_{k+1}|x_{k+1}})]$ is also nonlocal. By making use of Definition~\ref{def:nBellJM}, our main result may be reformulated as follows:
\begin{corollary} \label{thm:reformulation}
Let $\CC_2$ be the qubit Hilbert space and let $\oM = \{M_{a|x}\}$ be a finite set of incompatible POVMs. Then $\lim_{k\rightarrow \infty} \etaB{k}(\oM) = \etaJM(\oM)$.
\end{corollary}

We now present some examples to illustrate the definition of noisy $k$-partite Bell joint measurability. Let us start with the set of POVM $\oP={P_{a|x}}$ corresponding to the Pauli measurements $Z$ and $X$. More precisely, we define $P_{0|0}:= \ketbra{0}$, $P_{1|0}:=\ketbra{1}$, $P_{0|1}:= \ketbra{+}$, $P_{1|1}:=\ketbra{-}$,where $\ket{0}, \ket{1}, \ket{+}, \ket{-}$ are the eigenvectors of the Pauli $Z$ and $X$ matrices, repectively. The white noisy robustness for a pair of Pauli measurements is $\eta_{JM}(\oP)=1/\sqrt{2}$~\cite{HeinosaariMiyaderaZiman-compatibility}, we prove that $\etaB{2}(\oP) = \frac{1}{\sqrt[4]{2}}\approx 0.8409$ in the Appendix~\ref{app:XZ}. Note that the two values do not coincide, since here Alice and Bob have to perform the same set of measurements, contrary to Ref.~\cite{WolfPerezgarciaFernandez-measIncomp}.

Using similar ideas, we can also show that $\etaB{3}(\oP) \approx 0.7938$, proving that for a pair of Pauli measurements, $\etaB{2}(\oP)>\etaB{3}(\oP)$. To attain this goal, we first notice that we know all the tight Bell inequalities in tripartite scenario where all parties perform two dichotomic measurements~\cite{Sliwa-BellIneq,RossetbancalGisin-Bellnonlocality, werner2001all, zukowski2002bell}. Then, we fix our measurements as $\eta$ white noisy pair of Pauli measurements and evaluate the eigenvalues of all possible Bell operators.

We now focus our attention to the set of so-called trine measurements presented in Ref.~\cite{BeneVertesi-noBell}, which can be defined as $T_{a|x} := \frac{1}{2} (\I + (-1)^a \vec{v}_x \cdot \vec{\sigma})$ where $\vec{v}_x = \left[\cos\left(\frac{2 \pi x}{3}\right), 0, \sin\left(\frac{2 \pi x}{3}\right)\right]$, and $\vec{v}_x \cdot \vec{\sigma}:=\cos\left(\frac{2 \pi x}{3}\right)X+\sin\left(\frac{2 \pi x}{3}\right)Z$, with $X,Z$ being Pauli matrices. Ref.~\cite{BeneVertesi-noBell} shows $\etaJM(\oT) = \frac{2}{3}$ and for $\eta^* = 0.67>\frac{2}{3}$ the set $\oT^{\eta*} = D_\eta(T_{a|x})$ is incompatible, but $\Tr[\rho_2 (T^{\eta*}_{a|x} \otimes B_{b|y})]$ is Bell local for all $\rho_2 \in \dens(\CC_2 \otimes \CC_2)$ and all sets of POVMs $\{B_{b|y}\}$. Following our main result, $\oT^{\eta*}$ is a set of measurements whose nonlocality can be superactivated by considering multipartite scenarios.

Up to relabelling, the bipartite Bell scenario with three inputs, and two outputs has only two tight Bell inequalities, CHSH inequality and the I3322 inequality \cite{Froissart-BellIneq,CollinsGisin-BellIneq}. In order to use the CHSH inequality, we only make use two out of the three measurements of the trine given by $\{T_{a|x}^\eta\}$, and direct calculation shows that the maximal eigenvalue of its associated CHSH operator is $\eta^2\sqrt{7}$.
Hence, we see that $\eta^2 \sqrt{7}>2$ iff $\eta > \frac{\sqrt{2}}{\sqrt[4]{7}} \approx 0.8694$. We can show that the I3322 inequality and its relabelling are not violated for any state $\rho$ by setting $\eta = \frac{\sqrt{2}}{\sqrt[4]{7}}$ and evaluating the eigenvalues of its associated Bell operator. This shows that $\etaB{2}(\oT) = \frac{\sqrt{2}}{\sqrt[4]{7}}$.

By making use of a computational method detailed in the Appendix~\ref{app:trine-3parties}, we could also show that when three parties are considerd, we have that $\etaB{3}(\oT)<0.8007$, hence $\etaB{3}(\oT)<\etaB{2}(\oT)$. Finally, we remark that using the methods for symmetric multipartite Bell inequalities presented in Ref.~\cite{Designolle2024Mulatipartite}, one may cerfify nonlocality in scenarios with a large number of parties $N$, hence to obtain non-trivial (numerical and analytical) lower bounds for $\etaB{N}(\mathsf{M})$ for large $N$.

\paragraph{Implications for multipartite Bell nonlocality}
It is known that the set of all qubit POVMs is jointly measurable if and only if their white noisy visibility is smaller than or equals one have,  that is, if $\oM$ be the set of all qubit POVMs, $\eta^{\text{JM}}(\oM)=1/2$~\cite{Zhang2023POVMonehalf,Renner2023POVMonehalf}. Hence, by our main theorem, it follows that, if $\eta>1/2$, there exists a $N-$partite qubit state $\rho$ and a set of measurements $\{M_{a|x}\}$ such that 
$\Tr\big[\rho D_\eta(M_{a_1|x_1}) \otimes \ldots D_\eta(M_{a_N|x_N}) \big] = \Tr\big[D_\eta^{\otimes N}(\rho) \; M_{a_1|x_1} \otimes \ldots M_{a_N|x_N} \big]$ is Bell nonlocal.
\begin{corollary}\label{cor:1/2}
If $\eta> \frac{1}{2}$, there exists a $N\in\NN$, and a $N$ qubit state $\rho\in\dens(\CC_2^{\otimes N})$  such that the state $D_\eta^{\otimes N}(\rho)$ is entangled and violates a Bell inequality.
\end{corollary}
This demonstrates that there must be quantum states whose Bell nonlocality is very robust against noise. Note that GHZ states, if the standard Mermin inequality \cite{mermin1990extreme} is used, lead to a robustness corresponding to $\eta> \frac{1}{\sqrt{2}}$ only. So, this predicted effect is stronger than the exponential violation of local realism from the Mermin inequality.

\paragraph{Discussion and open questions}
We have proven a very strong link between qubit measurement compatibility and multipartite Bell scenarios. More precisely, we have shown that all sets of incompatible measurements lead to Bell nonlocality if sufficiently many parties are considered. For this, it is enough that all parties perform exactly the same measurement. When contrasted to the bipartite case~\cite{HirschQuintinoBrunner-noBell,BeneVertesi-noBell}, our results show that the nonlocality of incompatible measurements may be superactivated by considering multipartite scenarios. Our main result induced a novel quantifier for the Bell nonlocality and the incompatibility of POVMs, based on how many parties are required to obtain Bell nonlocality, a concept which was also explored in this work.

An immediate open question is to identify the multiparticle quantum states $\rho$ that are used to prove Theorem~\ref{thm:main}. Can they be restricted to specific known families of states, such as GHZ states or Dicke states? Similarly, Corollary~\ref{cor:1/2} proves the existence of multi-qubit states with an extreme robust violation of local realism. Can these states be identified and used as a resource for information processing tasks? Or can we at least bound the number of parties $N$?

A more challenging open question would be to generalize our results beyond qubits without the assumption of non-existence of NPT bound entangled states, but this may turn out to be an extremely tough problem to solve. But more importantly, it was recently observed that even an arbitrary small amount of Bell nonlocality is a resource for device-independent quantum key distribution \cite{wooltorton2024device}. One can thus ask whether any incompatible measurements on qubits (or even qudits) can be used for some form of multipartite device-independent quantum key distribution. If this turns out to be the case, we are likely to obtain device-independent quantum key distribution protocol for many parties with high robustness to local noise.

\paragraph{Acknowledgments}
We acknowledge Emmanuel Zambrini Cruzeiro and Flavien Hirsch, for fruitful discussions and Sébastien Designolle and Tamás Vértesi for comments in our numerical results and for pointing out references \cite{Designolle2024Mulatipartite}, and we acknowledge support from the Deutsche Forschungsgemeinschaft (DFG, German Research Foundation, project numbers 447948357 and 440958198), the Sino-German Center for Research Promotion (Project M-0294), the German Ministry of Education and Research (Project QuKuK, BMBF Grant No. 16KIS1618K), the DAAD, the Alexander von Humboldt Foundation, and the the Niedersächsisches Ministerium für Wissenschaft und Kultur.


\onecolumngrid
\appendix

\numberwithin{proposition}{section}
\numberwithin{lemma}{section}

\section{Proof of the main theorem: Details for step (a)} \label{app:obs1}
A conditional probability is a set of probability distributions, usually represented as $p(a|x)$ where $a$ and $x$ are suitable indexes, we will always assume that both $a$ and $x$ belong to finite sets. Here $x$ indexes the different probability distributions and $a$ indexes elements of said probability distribution, thus we get the normalization condition $\sum_a p(a|x) = 1$ for all $x$. It is straightforward to see that conditional probability distributions for a compact convex set, thus they can be seen as a state space of a suitable generalized probabilistic theory (GPT) \cite{Plavala-review}, in fact this is exactly the boxworld GPT \cite{Barrett-GPTinformation}. We will denote this GPT and its state space by $B(n_a, n_x)$ where $n_a$ is the number of admissible values of $a$ and $n_x$ is the number of admissible values of $x$. Here and from now on we will, without the loss of generality, assume that all probability distributions contained in the conditional probability distribution $p(a|x)$ contain the same number of outcomes. Moreover we will write $\{p(a|x)\}$ when necessary when referring to the behaviour as a whole rather than to the specific number $p(a|x)$ corresponding to the specific indexes $a,x$, one can understand $\{p(a|x)\}$ as a matrix with rows labeled by $a$ and columns labeled by $x$, while $p(a|x)$ as the matrix elements of $\{p(a|x)\}$.

A behavior is a multipartite non-signaling conditional probability: multipartite since behavior is a set of probability distributions between two or more parties, thus it is represented as $p(ab|xy)$ for two parties, or in general $p(a_1 \ldots a_N| x_1 \ldots x_N)$ for $N$ parties. We will use the vector notation $p(a_1 \ldots a_N| x_1 \ldots x_N) = p(\vec{a}|\vec{x})$ in order to shorten the notation when possible. One can also show that behaviors correspond to multipartite states of the boxworld GPT.

Behavior is non-signaling meaning that it has well-defined marginals that are obtained by summing over the probability distribution of one of the parties, thus we require that
\begin{equation}
\sum_{a_i} p(a_1 \ldots a_N| x_1 \ldots x_i \ldots x_N) = \sum_{a_i} p(a_1 \ldots a_N| x_1 \ldots x'_i \ldots x_N)
\end{equation}
for all $x_i, x'_i$ and for all $i$.

The behavior $p(a_1 \ldots a_N| x_1 \ldots x_N)$ is local if it can be written as
\begin{equation} \label{eq:obs1-local}
p(a_1 \ldots a_N| x_1 \ldots x_N) = \sum_\lambda p(\lambda) p^{(1)}_\lambda(a_1|x_1) \ldots p^{(N)}_\lambda(a_N|x_N)
\end{equation}
where $p(\lambda)$ is some probability distribution and $p^{(i)}_\lambda(a_i|x_i)$ are conditional probabilities, $i \in \{1, \ldots, N\}$. Note that most literature uses the notation $p^{(i)}(a_i|x_i, \lambda)$ instead of $p^{(i)}_\lambda(a_i|x_i)$, but this is a mere notation change and the later notation is better suited for our purposes. It follows from \eqref{eq:obs1-local} by definition that a behavior is local if and only if it can be seen as a separable state in the tensor product of the respective cones and GPTs. We thus get:
\begin{proposition} \label{prop:obs1-sep}
A behavior $p(a_1 \ldots a_N| x_1 \ldots x_N)$ is Bell local if and only if it can be seen as separable element of the tensor product of the underlying cones.
\end{proposition}
\begin{proof}
Here we state the formal proof; the result follows almost immediately from \eqref{eq:obs1-local}. Given a conditional probability $\{p(a|x)\}$ we can see the linear map $f_{a|x}: \{p(a|x)\}$ as element of the cone dual to the cone of conditional probabilities. We will use this insight to prove that \eqref{eq:obs1-local} holds if and only if
\begin{equation} \label{eq:obs1-localTensor}
\{p(a_1 \ldots a_N| x_1 \ldots x_N)\} = \sum_\lambda p(\lambda) \{p^{(1)}_\lambda(a_1|x_1)\} \otimes \ldots \otimes \{p^{(N)}_\lambda(a_N|x_N)\},
\end{equation}
the result follows from
\begin{align}
(f_{a_1|x_1} \otimes \ldots \otimes f_{a_N|x_N})(\{p(a_1 \ldots a_N| x_1 \ldots x_N)\}) &= p(a_1 \ldots a_N| x_1 \ldots x_N) \\
(f_{a_1|x_1} \otimes \ldots \otimes f_{a_N|x_N})(\{p^{(1)}_\lambda(a_1|x_1)\} \otimes \ldots \otimes \{p^{(N)}_\lambda(a_N|x_N)\}) &= p^{(1)}_\lambda(a_1|x_1) \ldots p^{(N)}_\lambda(a_N|x_N)
\end{align}
and from the fact that $f_{a_1|x_1} \otimes \ldots \otimes f_{a_N|x_N}$ contains a basis of the dual vector space. This immediately yields that \eqref{eq:obs1-local} follows from \eqref{eq:obs1-localTensor}.

If \eqref{eq:obs1-local} holds, then we can construct the sum $\sum_\lambda p(\lambda) \{p^{(1)}_\lambda(a_1|x_1)\} \otimes \ldots \otimes \{p^{(N)}_\lambda(a_N|x_N)\}$ and \eqref{eq:obs1-localTensor} must hold simply because the two vectors yield the same value when evaluated on all possible $f_{a_1|x_1} \otimes \ldots \otimes f_{a_N|x_N}$.
\end{proof}

Any behaviour that is not of the form given by equation \eqref{eq:obs1-local} is called nonlocal. Given a nonlocal behaviour $p(\vec{a}|\vec{x})$ there exists a linear functional $W$ such that $W(\{p(\vec{a}|\vec{x})\}) > 0$ but such that $W$ is upper bounded by $0$ for all local behaviours; the existence of $W$ follows from the hyperplane separation theorem \cite[Section 2.5.1]{BoydVandenberghe-convex} since the set of local behaviours is convex and closed. One can understand $W$ either as an entanglement witness in the boxworld GPT, or as a Bell inequality. It thus follows that for every nonlocal behaviour $p(\vec{a}|\vec{x})$ there is some Bell inequality which is violated by $p(\vec{a}|\vec{x})$; this is in fact a known result \cite{BrunnerCavalcantiPironioScaraniWehner-BellNonlocality}.

\section{Proof of the main theorem: Details for steps (b) and (c)} \label{app:obs2}
In this section we will first introduce the concept of separable map, and then we will prove two lemmas that will be instrumental in the final proof, the main result of this section is Proposition~\ref{prop:obs2-sep}. The proof technique we use in this section is based on the proof of \cite[Lemma 2.8]{AubrunMullerhermes-anihilation}.

In this section we will be dealing with linear maps of the form $\Psi_1: \lin(\CC_d) \to \linspan(B(n_a, n_x))$ and $\Psi_2: \linspan(B(n_a, n_x)) \to \lin(\CC_d)$. Here, $B(n_a, n_x)$ denote the set of behaviours with $n_x$ measurements and $n_a$ outcomes. We will say that $\Psi_1$ is positive if for every $\rho \in \dens(\CC_d)$ we have $\Psi_1(\rho) \in B(n_a, n_x)$, in this case we also write $\Psi_1: \dens(\CC_d) \to B(n_a, n_x)$ and we say that $\Psi_2$ is positive if for every $\{p(a|x)\} \in B(n_a, n_x)$ we have $\Psi_2(\{p(a|x)\}) \in \dens(\CC_d)$ and we write $\Psi_2: B(n_a, n_x) \to \dens(\CC_d)$. These notions of positivity are based on the same idea as the standard notion of positivity of quantum maps, that is, to saying that a linear map $\Psi_3: \lin(\CC_d) \to \lin(\CC_d)$ is positive if for very $\rho \in \dens(\CC_d)$ we have $\Psi_3(\rho) \in \dens(\CC_d)$, but since $B(n_a, n_x) \neq \dens(\CC_d)$, these notions are strictly different.

Let $\oM = \{M_{a|x}\}$ be a finite set of POVMs. We will, without the loss of generality, assume that all POVMs in the given collection contain outcomes $n_a$ outcomes and that $\oM$ consists of $n_x$ POVMs. Given a state $\rho \in \dens(\CC_d)$, we can measure $\rho$ using $\oM$ and obtain the conditional probability $p(a|x) = \Tr(M_{a|x} \rho)$. Thus $\oM$ is a positive linear map from $\dens(\CC_d)$ to $B(n_a, n_x)$.

Using  known theory of linear maps \cite{Ryan-tensorProducts}, one can see that $\oM$ corresponds to a vector in the tensor product vector space $\lin(\CC_d) \otimes \linspan(B(n_a, n_x))$. This can be seen as a generalized version of the Choi-Jamio{\l}kowski isomorphism.

Indeed, there are some operators $F_\alpha \in \lin(\CC_d)$ and conditional probability distributions $p_\alpha(a|x)$ indexed by a suitable index $\alpha$ such that
\begin{equation} \label{eq:obs2-Mtensor}
M_{a|x} = \sum_{\alpha} F_\alpha p_\alpha(a|x)
\end{equation}
and we have $\Tr(M_{a|x} \rho) = \sum_\alpha \Tr(F_\alpha \rho) p_\alpha(a|x)$ for $\rho \in \dens(\CC_d)$. We have omitted the tensor product in \eqref{eq:obs2-Mtensor}, i.e., we use $F_\alpha p_\alpha(a|x)$ instead of $F_\alpha \otimes p_\alpha(a|x)$, since if $a, x$ are fixed, then $p_\alpha(a|x)$ is just a number. Note that $F_\alpha$ are in general not positive semidefinite. For a given set of POVMs one can find the representation given by \eqref{eq:obs2-Mtensor} as follows: fix $\{p_\alpha(a|x)\}$ that form a linear basis of $\linspan(B(n_a, n_x))$, then for any $\rho \in \dens(\CC_d)$ we have $\Tr(M_{a|x} \rho) = \sum_\alpha f_\alpha(\rho) p_\alpha(a|x)$. Since $f_\alpha(\rho)$ are linear functions of $\rho$, it follows that there are operators $F_\alpha$ such that $f_\alpha(\rho) = \Tr(F_\alpha \rho)$.

A special case of $\oM$ is when the corresponding vector in the tensor product vector space is separable, that is when we have
\begin{equation} \label{eq:obs2-Mseparable}
M_{a|x} = \sum_{\lambda} G_\lambda p_\lambda(a|x),
\end{equation}
where now $G_\lambda$ are positive semidefinite. If $\oM = \{M_{a|x}\}$ is of the form given by \eqref{eq:obs2-Mseparable} then we say that it is a separable map. We will now show that if $\oM$ is {\it not} a separable map, then it gives rise to a positive but {\it not} completely positive map $\Phi: \lin(\CC_d) \to \lin(\CC_d)$.

\begin{lemma} \label{lemma:obs2-Phi-PnotCP}
Let $\oM$ be a set of POVMs such that $\oM$ is not separable map, then there is positive map $\tilde{Q}: B(n_a, n_x) \to \lin(\CC_d)$ such that the map $\Phi = \tilde{Q} \circ \oM: \lin(\CC_d) \to \lin(\CC_d)$ is positive but not completely positive.
\end{lemma}
\begin{proof}
Since convex combination of separable vectors is again separable, it is straightforward to check that the set of collections of POVMs with $n_a$ outcomes and $n_x$ measurements that are separable maps is convex. So, by the hyperplane separation theorem \cite[Section 2.5.1]{BoydVandenberghe-convex} it follows that there is a witness-type operator. Concretely, if $\oM$ is not a separable map, then there exists a linear functional $Q$ on $\lin(\CC_d) \otimes \linspan(B(n_a, n_x))$ such that $Q$ is positive on all collections of POVMs that are separable maps and $Q(\oM) < 0$. Since all of the vector spaces in question are finite-dimensional, it follows that we can express $Q$ as
\begin{equation} \label{eq:obs2-Qtensor}
Q = \sum_\beta X_\beta \otimes W_\beta
\end{equation}
where $X_\beta$ are linear functionals on $\lin(\CC_d)$ and $W_\beta$ are linear functionals on conditional probabilities. Moreover using the self-duality of $\lin(\CC_d)$ we get $X_\beta \in \lin(\CC_d)$. We obtain
\begin{equation} \label{eq:obs2-QonM}
0 > Q(\oM) = \sum_{\alpha, \beta} \Tr(F_\alpha X_\beta) W_\beta(\{p_\alpha(a|x)\})
\end{equation}
where we have used the decomposition of $\oM$ as given by \eqref{eq:obs2-Mtensor}. Note that the final result does not depend on $a, x$ because the functionals $W_\beta$ acts on the space of conditional probabilities.

Note that one can also see $Q$ as a positive map $\tilde{Q}: B(n_a, n_x) \to \lin(\CC_d)$ given as
\begin{equation}
\tilde{Q}(\{p(a|x)\}) = \sum_\beta W_\beta(\{p(a|x)\}) X_\beta.
\end{equation}
To check that $\tilde Q$ is positive note that for $E \geq 0$ we have
\begin{equation}
\Tr(\tilde{Q}(\{p(a|x)\}) E) = Q(\{p(a|x)\} \otimes E) \geq 0
\end{equation}
where we have used that $Q$ is positive on separable maps.

Given $\oM$ and $Q$ we can construct the map $\Phi: \lin(\CC_d) \to \lin(\CC_d)$ given as $\Phi = \tilde{Q} \circ \oM$, i.e., $\Phi(\rho) = \tilde{Q}(\oM(\rho))$ for $\rho \in \dens(\CC_d)$. Equivalently, using \eqref{eq:obs2-Mtensor} and \eqref{eq:obs2-Qtensor}, one can construct $\Phi$ as a linear map corresponding to the tensor $\sum_{\alpha, \beta} W_\beta(\{p_\alpha(a|x)\}) F_\alpha \otimes X_\beta$, that is, for $\rho \in \dens(\CC_d)$ we have
\begin{equation}
\Phi(\rho) = \sum_{\alpha, \beta} W_\beta(\{p_\alpha(a|x)\}) \Tr(F_\alpha \rho) X_\beta.
\end{equation}
We will now prove that $\Phi$ is positive, but not completely positive map. One can already see that $\Phi$ is positive since $\oM$ and $\tilde{Q}$ are positive, but nevertheless we provide an explicit proof. To see that $\Phi$ is positive, note that for $E \geq 0$ we have
\begin{align}
\Tr(\Phi(\rho) E) &= \sum_{\beta} W_\beta\left( \sum_\alpha \{p_\alpha(a|x)\} \Tr(F_\alpha \rho) \right) \Tr(X_\beta E) \\
&= \sum_{\beta} W_\beta \left( \{\Tr(M_{a|x} \rho)\} \right) \Tr(X_\beta E) \\
&= Q(\{\Tr(M_{a|x} \rho)\} \otimes E)  \geq 0
\end{align}
simply because $\Tr(M_{a|x} \rho)$ is a valid conditional probability and $Q$ is positive on all separable maps.

To prove that $\Phi$ is not completely positive, we will prove that the Choi matrix of $\Phi$ is not positive semidefinite. Let $\ket{i}$ be an orthonormal basis of $\CC_d$ and let $\ket{\phi^+} = \frac{1}{d} \sum_i \ket{ii}$ be the maximally-entangled state. We then define the Choi matrix of $\Phi$ as $\J(\Phi) = (\Phi \otimes \id)(\ketbra{\phi^+})$, we have
\begin{equation}
\J(\Phi) = \sum_{\alpha, \beta} W_\beta(\{p_\alpha(a|x)\}) X_\beta \otimes \Tr_1((F_\alpha \otimes \I) \ketbra{\phi^+}).
\end{equation}
and
\begin{align}
\dim(\CC_d)^2 \Tr(\J(\Phi) \ketbra{\phi^+}) &= \dim(\CC_d)^2 \sum_{\alpha, \beta} W_\beta(\{p(a|x,\alpha)\}) \Tr( (X_\beta \otimes \Tr_1((F_\alpha \otimes \I) \ketbra{\phi^+})) \ketbra{\phi^+}) \\
&= \dim(\CC_d) \sum_{\alpha, \beta} W_\beta(\{p(a|x,\alpha)\}) \Tr(X_\beta^\intercal \Tr_1((F_\alpha \otimes \I) \ketbra{\phi^+})) \\
&= \dim(\CC_d) \sum_{\alpha, \beta} W_\beta(\{p(a|x,\alpha)\}) \Tr((F_\alpha \otimes X_\beta^\intercal) \ketbra{\phi^+})) \\
&= \sum_{\alpha, \beta} W_\beta(\{p(a|x,\alpha)\}) \Tr(F_\alpha X_\beta) < 0,
\end{align}
where we have used \eqref{eq:obs2-QonM} and $\dim(\CC_d) \Tr((A \otimes B) \ketbra{\phi^+}) = \Tr(A^\intercal B)$ where $A^\intercal$ is the transpose of $A$. This shows that $\J(\Phi)$ is not positive semidefinite, thus $\Phi$ is not completely positive.
\end{proof}

Given a finite set of POVMs $\oM$ and a multipartite quantum state $\rho_N \in \dens(\CC_d^{\otimes N})$, $N \in \NN$, we can measure $M_{a|x}$ at every local Hilbert space, obtaining the behavior $\Tr( (\otimes_{i=1}^N M_{a_i|x_i}) \rho_N)$. We will now prove that if this behavior is always local then for any positive linear map $\tilde{R}: B(n_a, n_x) \to \lin(\CC_d)$ the linear map $\Psi = \tilde{R} \circ \oM$ is positive and locally entanglement annihilating \cite{MoravcikovaZiman-LEA}.
\begin{lemma} \label{lemma:obs2-Phi-LEA}
Let $\oM$ be a set of POVMs such that $\oM^{\otimes N}(\rho_N)$ is separable behavior for all $N \in \NN$ and all $\rho_N \in \dens(\CC_d^{\otimes N})$. Then for any positive linear map $\tilde{R}: B(n_a, n_x) \to \lin(\CC_d)$ we have that $\Psi = \tilde{R} \circ \oM$ is locally entanglement annihilating, i.e., $\Psi^{\otimes N}(\rho_N) \geq 0$ and $\Psi^{\otimes N}(\rho_N)$ is separable quantum state for all $N \in \NN$ and all $\rho_N \in \dens(\CC_d^{\otimes N})$.
\end{lemma}
\begin{proof}
Assume that for every $N \in \NN$ and every state $\rho_N \in \dens(\CC_2^{\otimes N})$ the behavior $\Tr( (\otimes_{i=1}^N M_{a_i|x_i}) \rho_N)$ is local, i.e., separable, so we have
\begin{equation} \label{eq:obs2-Phi-LEA-sepElement}
\Tr( (\otimes_{i=1}^N M_{a_i|x_i}) \rho_N) = \sum_\lambda p(\lambda) p^{(1)}_\lambda(a_1|x_1) \ldots p^{(N)}_\lambda(a_N|x_N)
\end{equation}
which is the same as
\begin{equation} \label{eq:obs2-Phi-LEA-sepTensor}
\oM^{\otimes N}(\rho_N) = \sum_\lambda p(\lambda) \{p^{(1)}_\lambda(a_1|x_1)\} \otimes \ldots \otimes \{p^{(N)}_\lambda(a_N|x_N)\}
\end{equation}
since \eqref{eq:obs2-Phi-LEA-sepTensor} is a matrix equality while \eqref{eq:obs2-Phi-LEA-sepElement} is equality all entries of the corresponding matrices. We then have
\begin{align}
\Psi^{\otimes N}(\rho_N) &= \tilde{R}^{\otimes N}(\oM^{\otimes N}(\rho_N)) \\
&= \sum_\lambda p(\lambda) \tilde{R}(\{p^{(1)}_\lambda(a_1|x_1)\}) \otimes \ldots \otimes \tilde{R}(\{p^{(N)}_\lambda(a_N|x_N)\})
\end{align}
and we get that $\Psi^{\otimes N}(\rho_N)$ is separable since $\tilde{R}$ is positive and thus we have $\tilde{R}(\{p^{(i)}_\lambda(a_i|x_i)\}) \in \dens(\CC_d)$.
\end{proof}

We are now prepared to state the main result of this section:
\begin{proposition} \label{prop:obs2-sep}
Let $\dim(\CC_2) = 2$. Let $\oM = \{M_{a|x}\}$ be a finite set of POVMs on $\CC_2$, so $M_{a|x} \in \lin(\CC_2)$. If $\oM^{\otimes N}(\rho_N)$ is separable behavior for all $N \in \NN$ and all $\rho_N \in \dens(\CC_d^{\otimes N})$ then $\oM$ is a separable map.
\end{proposition}
\begin{proof}
We prove the result by contradiction: assume that $\oM$ is not a separable map such that $\oM^{\otimes N}(\rho_N)$ is separable behavior for all $N \in \NN$ and all $\rho_N \in \dens(\CC_d^{\otimes N})$. Then, according to Lemma~\ref{lemma:obs2-Phi-PnotCP}, there is a positive but not completely positive map $\Phi = \tilde{Q} \circ \oM$, where $\tilde{Q}: B(n_a, n_x) \to \dens(\CC_2)$ is a positive linear map and according to Lemma~\ref{lemma:obs2-Phi-LEA} the aforementioned map $\Phi$ is locally entanglement annihilating. Thus we get that $\Phi: \dens(\CC_2) \to \dens(\CC_2)$ is positive, not completely positive, and locally entanglement annihilating. This is a contradiction with Lemma~\ref{lemma:obs2-qubit-LeaEB} below, which states any positive and locally entanglement annihilating map $\Phi: \dens(\CC_2) \to \dens(\CC_2)$ is entanglement breaking, thus completely positive.
\end{proof}

\begin{lemma} \label{lemma:obs2-qubit-LeaEB}
Let $\dim(\CC_2) = 2$. If a linear map $\Phi: \lin(\CC_2) \to \lin(\CC_2)$ is positive and locally entanglement annihilating, that is, for every $N \in \NN$ and every $\rho_N \in \dens(\CC_2^{\otimes N})$ we have that $\Phi^{\otimes N}(\rho_N)$ is separable, then $\Phi$ is completely positive and entanglement breaking.
\end{lemma}
\begin{proof}
It was previously shown in \cite[Lemma 3.2]{ChristandlMullerhermesWolf-composedETB} that if $\Phi: \lin(\CC_2) \to \lin(\CC_2)$ is completely positive and locally entanglement annihilating, then $\Phi$ is entanglement breaking.

Assume that $\Phi$ is not completely positive. We still have that $\Phi^{\otimes N}(\rho_N)$ is separable for all $N \in \NN$ and all $\rho_N \in \dens(\CC_2^{\otimes N})$, which implies that $\Phi^{\otimes N}$ is a positive map for all $N \in \NN$. Such maps are called tensor-stable positive maps and $\Phi$ clearly is a tensor-stable positive. It was shown in \cite[Theorem 2]{MullerhermesReebWolf-tensorPositive} that any tensor-stable positive map on qubits is either completely positive or completely copositive, the map $\Psi$ is completely copositive if the map $\Psi^\intercal: \lin(\CC_2) \to \lin(\CC_2)$ given as $\Psi^\intercal(\rho) = (\Psi(\rho))^\intercal$ is completely positive, here $A^\intercal$ is the transpose of $A$. Since we have assumed that $\Phi$ is not completely positive, it follows that $\Phi$ must be completely copositive and so $\Phi^\intercal$ is completely positive, moreover it is straightforward to show that if $\Phi$ is locally entanglement annihilating, then so is $\Phi^\intercal$. Then it follows from \cite[Lemma 3.2]{ChristandlMullerhermesWolf-composedETB} that $\Phi^\intercal$ is entanglement breaking, but this implies that also $\Phi$ is entanglement breaking, since the transpose of entanglement breaking map is again entanglement breaking. Moreover $\Phi$ must be completely positive since any entanglement-breaking map is completely positive.
\end{proof}

\section{Proof of the main theorem: Details for step (d)} \label{app:obs3}
The following statement is almost immediate to prove:
\begin{proposition} \label{prop:obs3-sep}
If $\oM$ is a separable map, then the POVMs contained in $\oM$ are compatible.
\end{proposition}
\begin{proof}
Assume that $\oM$ is of the form as given by \eqref{eq:obs2-Mseparable}, by summing over $a$ we obtain $\I = \sum_a \sum_\lambda p_\lambda(a|x) G_\lambda = \sum_\lambda G_\lambda$ where we have used that $\sum_a p_\lambda(a|x) = 1$ for all $x$. If follows that $\{G_\lambda\}$ is a POVM. Moreover, changing the notation from $p_\lambda(a|x)$ to $p(a|x, \lambda)$, we get $M_{a|x} = \sum_\lambda p(a|x, \lambda) G_\lambda$, i.e., that all POVMs in $\oM$ are a post-processing of $\{G_\lambda\}$, the post-processings are given by $p(a|x, \lambda)$. This is one of the equivalent definitions of compatibility of POVMs \cite{HeinosaariMiyaderaZiman-compatibility,GuhneHaapasaloKraftPellonpaaUola-incomaptiblity}. It is straightforward to see that the converse follows as well.
\end{proof}

\section{Proof of Theorem~\ref*{thm:main}} \label{app:proof}
\thmMain*
\begin{proof}
Assume that for every $N \in \NN$ and for every state $\rho_H \in \dens(\CC_2^{\otimes N})$ we have that $p(a_1 \ldots a_N|x_1 \ldots x_N) = \Tr(\rho(M_{a_1|x_1} \otimes \ldots \otimes M_{a_N|x_N})$ is Bell local. Then, according to Prop.~\ref{prop:obs1-sep}, $\{p(a_1 \ldots a_N|x_1 \ldots x_N)\}$ is a separable element of the tensor product of the underlying cones. But, since this holds for all $N \in \NN$ and for all $\rho_H \in \dens(\CC_2^{\otimes N})$, we get from Prop.~\ref{prop:obs2-sep} that $\oM$ is a separable map. But then, according to Prop.~\ref{prop:obs3-sep}, the set of POVMs $\oM$ is compatible.

The statement of the theorem is obtained by contraposition: if $\oM$ contains incompatible POVMs, then there must exist some $N \in \NN$ and $\rho_H \in \dens(\CC_2^{\otimes N})$ such that $p(a_1 \ldots a_N|x_1 \ldots x_N) = \Tr(\rho(M_{a_1|x_1} \otimes \ldots \otimes M_{a_N|x_N})$ is Bell nonlocal and thus violates some Bell inequality.
\end{proof}

\section{Proof of Theorem~\ref*{thm:NPTboundEnt}} \label{app:proofNPTboundEnt}
\thmNPTboundEnt*
\begin{proof}
Assume that $\oM$ is incompatible, then, according to Lemma~\ref{lemma:obs2-Phi-PnotCP}, $\oM$ gives rise to $\Phi = \tilde{Q} \circ \oM: \lin(\CC_d) \to \lin(\CC_d)$ that is positive but not completely positive. Moreover since for every $N \in \NN$ and for every state $\rho_H \in \dens(\CC_d^{\otimes N})$ we have that $p(a_1 \ldots a_N|x_1 \ldots x_N) = \Tr(\rho(M_{a_1|x_1} \otimes \ldots \otimes M_{a_N|x_N})$ is Bell local, then according to Lemma~\ref{lemma:obs2-Phi-LEA} we have that $\Phi$ is locally entanglement annihilating.

Since $\Phi$ is not completely positive, it is not entanglement breaking. We have thus constructed linear map $\Phi: \lin(\CC_d) \to \lin(\CC_d)$ which is locally entanglement annihilating but not entanglement breaking. It was proved in \cite{AubrunMullerhermes-anihilation} that existence of such map implies the existence of bound entangled quantum states with non-positive partial transpose in $\dens(\CC_{d^2}^{\otimes 2})$. The result thus follows by contradiction.
\end{proof}

\section{Pauli $X$ and $Z$ measurements} \label{app:XZ}
In a bipartite scenario where each part has two dichotomic measurements, all tight Bell inequalities are equivalent to the CHSH inequality ~\cite{Fine-BellIneq}, hence we just need to check when noisy Pauli measurements may lead to a Bell violation of this inequality. We then consider that both Alice and Bob perform the set of two Pauli measurements, $\oP$ and that $P_x:=P_{0|x}-P_{1|x}$ is the observable associated to the POVM $\{P_{0|x},P_{1|x}\}$. The behaviour $p(a_1a_2|x_1x_2)=\Tr[\rho (P_{a_1|x_1}\otimes P_{a_2|x_2})]$ is Bell local for all states $\rho\in\dens(\CC_2\otimes\CC_2)$ if and only if the operator $\text{CHSH}:=P_{0}\otimes P_{0}+P_{0}\otimes P_{1}+P_{1}\otimes P_{0}-P_{1}\otimes P_{1}$ has an eigenvalue strictly greater than $2$~\footnote{Strictly speaking one should consider all $8$ CHSH inequalities which are equivalent via relabelling of inputs and outputs, but due to the symmetry of the measurements we are considering, analysing a single representant of the CHSH inequality suffices.}. Now notice that for noisy Pauli measurements $D_\eta(P_{a|x})$, the associated observable reads as $\eta P_{x}$ hence, the CHSH operator takes the form of $\eta^2\text{CHSH}$ which has its maximal eigenvalue given by $\eta^2 2\sqrt{2}$~\cite{chsh69,chsh69-erratum,tsirelson85}, a quantity which is greater than two iff $\eta^2>\frac{1}{\sqrt{2}}$, hence $\etaB{2}(\oP) = \frac{1}{\sqrt[4]{2}}\approx 0.8409$.

\section{Trine measurements} \label{app:trine-3parties}
The full list of Bell inequalities for scenarios with three parties, three inputs, and two outputs is not known and finding these Bell inequalities is computationally challenging~\cite{RossetbancalGisin-Bellnonlocality}. Since the Bell inequalities are unknown, finding the optimal state $\rho$ is not an obvious task. We can however obtain an upper bound on $\etaB{3}(\oT)$, for that it is enough to find a tripartite quantum state $\rho$ such that $\Tr[\rho (T_{a|x}^\eta\otimes T_{b|y}^\eta\otimes T_{c|z}^\eta)]$, and verifying if a behaviour is Bell local can be done by linear programming~\cite{BrunnerCavalcantiPironioScaraniWehner-BellNonlocality}. By choosing the tripartite state $\rho_3:=\ketbra{\text{GHZ}_Y}$, where $\ket{\text{GHZ}_Y}:=\frac{\ket{+_y}\ket{+_y}\ket{+_y}+\ket{-_y}\ket{-_y}\ket{-_y}}{\sqrt{2}}$ and $\ket{\pm_y}:=\frac{\ket{0}\pm i\ket{1}}{\sqrt{2}}$ are the eigenvectors of the Pauli $Y$ operator, linear programming shows that for $\eta=0.8007$, the behaviour $ \Tr[\rho_3 (T^{\eta}_{a|x} \otimes T^{\eta}_{b|y} \otimes T^{\eta}_{c|z})] $ is Bell NL, hence $\etaB{3}(\oT)<0.8007$, and we have that $\etaB{3}(\oT)<\etaB{2}(\oT)$. All computational code used here is available at~\cite{MTQ_github}. Additionally, in the case of three parties and the set of POVMs $\oT^{\eta*}$, one can prove that for any state $\rho_3 \in \dens(\CC_2^{\otimes 3})$ we have $\Tr[\rho_3 (T^{\eta^*}_{a|x} \otimes T^{\eta^*}_{b|y} \otimes T^{\eta^*}_{c|z})] = p^3 \Tr[\rho_3 (T_{a|x} \otimes T_{b|y} \otimes T_{c|z})] + (1-p^3) p_L(abc|xyz)$ where $p_L(abc|xyz)$ is a behaviour with local hidden variable model and $p \approx 10^{-2}$, see the Appendix~\ref{app:trine-upperBound} for a proof. Hence, it is likely that $\etaB{3}(\oT^{\eta^*}) = 1$, and we need more parties to superactivate the nonlocality of $\oT^{\eta^*}$.

\section{Upper bound on the Bell nonlocality of the trine measurement for three parties} \label{app:trine-upperBound}
We will use that $T^{\eta*}_{a|x} = p T_{a|x} + (1-p) T^{\frac{2}{3}}_{a|x}$ where $T^{\frac{2}{3}}_{a|x} = \frac{2}{3} T_{a|x} + \frac{\I}{6}$ and $p = \dfrac{\eta^* - \eta_3}{1 - \eta_3} \approx 10^{-2}$. Let $\rho_3 \in \dens(\CC_2 \otimes \CC_2 \otimes \CC_2)$, we then have
\begin{equation}
\begin{split}
\Tr(\rho_3(T^{\eta^*}_{a|x} \otimes T^{\eta^*}_{b|y} \otimes T^{\eta^*}_{c|z})) &=
p^3 \Tr(\rho_3(T_{a|x} \otimes T_{b|y} \otimes T_{b|y})) \\
&+ p (1-p) \Tr(\rho_3(T_{a|x} \otimes T^{\eta^*}_{b|y} \otimes T^{\frac{2}{3}}_{c|z} + T^{\eta^*}_{a|x} \otimes T^{\frac{2}{3}}_{b|y} \otimes T_{c|z} + T^{\frac{2}{3}}_{a|x} \otimes T_{b|y} \otimes T^{\eta^*}_{c|z})) \\
&+ (1-p)^3 \Tr(\rho_3(T^{\frac{2}{3}}_{a|x} \otimes T^{\frac{2}{3}}_{b|y} \otimes T^{\frac{2}{3}}_{c|z})).
\end{split}
\end{equation}
We will now show that terms of the form $\Tr(\rho_3(T_{a|x} \otimes T^{\eta^*}_{b|y} \otimes T^{\frac{2}{3}}_{c|z}))$ are Bell local. Since $T^{\frac{2}{3}}_{c|z}$ are compatible, there is a POVM $G_\lambda$ and conditional probability $p^{(3)}(c|z, \lambda)$ such that $T^{\frac{2}{3}}_{c|z} = \sum_\lambda p^{(3)}(c|z, \lambda) G_\lambda$. We thus get $\Tr(\rho_3(T_{a|x} \otimes T^{\eta^*}_{b|y} \otimes T^{\frac{2}{3}}_{c|z})) = \sum_\lambda p^{(3)}(c|z, \lambda) \Tr(\rho_3(T_{a|x} \otimes T^{\eta^*}_{b|y} \otimes G_\lambda))$. Denoting $\Tr_3[\rho_3(\I \otimes \I \otimes G_\lambda)] = p(\lambda) \sigma_\lambda$, where $\sigma_\lambda \in \dens(\CC_2 \otimes \CC_2)$ and $p(\lambda) = \Tr[\rho_3(\I \otimes \I \otimes G_\lambda)]$ is the corresponding probability, we get
\begin{equation}
\Tr(\rho_3(T_{a|x} \otimes T^{\eta^*}_{b|y} \otimes T^{\frac{2}{3}}_{c|z})) = \sum_\lambda p(\lambda) p^{(3)}(c|z, \lambda) \Tr(\sigma_\lambda(T_{a|x} \otimes T^{\eta^*}_{b|y})).
\end{equation}
Since it was proved in \cite{BeneVertesi-noBell} that $T^{\eta^*}_{b|y}$ does not violate any bipartite Bell inequality, no matter what is measured by the other party, we must have
\begin{equation}
 \Tr(\sigma_\lambda(T_{a|x} \otimes T^{\eta^*}_{b|y})) = \sum_\mu p(\mu | \lambda) p^{(1)}(a|x, \mu, \lambda) p^{(2)}(b|y, \mu, \lambda)
\end{equation}
and we get
\begin{align}
\Tr(\rho_3(T_{a|x} \otimes T^{\eta^*}_{b|y} \otimes T^{\frac{2}{3}}_{c|z})) &= \sum_\lambda p(\lambda) p^{(3)}(c|z, \lambda) \sum_\mu p(\mu | \lambda) p^{(1)}(a|x, \mu, \lambda) p^{(2)}(b|y, \mu, \lambda) \\
&= \sum_\lambda p(\mu, \lambda) p^{(1)}(a|x, \mu, \lambda) p^{(2)}(b|y, \mu, \lambda) p^{(3)}(c|z, \lambda)
\end{align}
which is a local hidden variable model. It thus follows that the only term in $p(a b c| x y z)$ that may be Bell nonlocal is $\Tr(\rho_3(T_{a|x} \otimes T_{b|y} \otimes T_{b|y}))$. But this term comes with the factor of $p^3 \approx 10^{-6}$ which will make the violation of any Bell inequality extremely small or easily nullified by the two other terms which are of order $p^1$ and $p^0$ respectively.

\bibliography{citations}

\end{document}